\documentclass{article}
\usepackage[utf8]{inputenc}
\usepackage{amsmath,amsthm,amssymb,fullpage,graphicx,url,xspace,cite,quotes,algorithm2e}
\usepackage[hidelinks]{hyperref}
\newtheorem{theorem}{Theorem}[section]

\title{Convex Cover and Hidden Set in Funnel Polygons}

\author{Reilly Browne\thanks{Department of Computer Science,
        Stony Brook University, {\tt rjbrowne@cs.stonybrook.edu}}}

\begin{document}
\thispagestyle{empty}
\maketitle

\begin{abstract}
We present linear-time algorithms for both maximum hidden set and minimum convex cover in funnel polygons. These algorithms show that funnel polygons are ``homestead'' polygons, i.e. polygons for which the hidden set number and the convex cover number coincide. We extend the algorithm to apply to maximum hidden vertex set and use the result to give a 2-approximation for all three problems in pseudotriangles.
\end{abstract}

\section{Introduction}
We study two classic visibility problems,\textit{ maximum hidden set} and \textit{minimum convex cover}. A hidden set is a set $H$ of points in a polygon $P$ of such that no two points $p,q\in H$ ``see'' each other, where $p$ sees $q$ if the line segment is in P, $pq \subseteq P$. A maximum hidden set is such a set of maximum cardinality, i.e. the problem is finding such a maximum set for a given polygon instance. A convex cover of a polygon $P$ is a set of convex polygons $C$ such that the union of all the polygons is equal to $P$, $\bigcup _{C_i \in C} = P$. Note that these polygons must lie within $P$. A minimum convex cover is such a set of minimum cardinality, i.e. the problem is finding such a minimum set for a given polygon. A relationship between the size of a maximum hidden set in a polygon $P$, $hs(P)$, and the size of a minimum convex cover of $P$, $cc(P)$, comes from Shermer \cite{Shermer}.

$$hs(P) \leq cc(P) $$

This means that for any polygon, it suffices to provide a hidden set of size $k$ and a convex cover of size $k$ to show that $k = hs(P) = cc(P) = k$. This is not always possible for all polygons, but for those which this is the case Browne and Chiu \cite{browne2022collapsing} classify these as ``homestead'' polygons.

Both maximum hidden set and minimum convex cover are APX-hard \cite{EIDENBENZHidden, EidenbenzConvexCover} when considering the case of general simple polygons, precluding a PTAS unless P=NP. However, this hardness result does not extend to many common subclasses of polygons. Browne and Chiu \cite{browne2022collapsing} gave linear-time algorithms for simultaneously calculating maximum hidden set and minimum convex cover in spiral polygons and histogram polygons. There also exists polynomial-time for calculating more restricted versions of these problems in polygon subclasses. Notably, Ghosh, Maheshwari, Pal, Saluja and Madhavan's \cite{GhoshWeakVisible} algorithm  finds the hidden set restricted to vertices in polygons that are weakly visible from a convex edge. This algorithm runs in $O(n^2)$ time where $n$ is the number of vertices. While no explicit algorithm is given for funnel polygons, Choi, Shin, and Chwa \cite{choi1995characterizing} characterize the visibility graph (of the vertices) of funnel polygons as weakly triangulated graphs. Hidden vertex set is equivalent to independent set on the visibility graph of a polygon. Since weakly triangulated graphs are perfect graphs, independent set can be found in polynomial-time \cite{GROTSCHEL1984325}. Using the specific properties of weakly triangulated graphs, Spinrad and Sritharan's \cite{spinrad1995algorithms} algorithm for independent set implies a faster algorithm for hidden vertex set in funnel polygons, which takes $O(n^4)$ time. As funnel polygons are a subclass of the more general class of polygons weakly visible from a convex edge, Ghosh, Maheshwari, Pal, Saluja and Madhavan's \cite{GhoshWeakVisible} algorithm can also be used, achieving an even faster time of $O(n^2)$.

In Section \ref{sec:funnel}, we describe an algorithm for solving both maximum hidden set and minimum convex cover in funnel polygons, and in Section \ref{sec:variants} we describe how the algorithm can be adapted to the problem of finding the maximum hidden vertex set. We extend these techniques in Section \ref{sec:pseudos} to the more general case of pseudotriangles, achieving a simple 2-approximation. It is currently not known if either problem is NP-hard for the case of pseudotriangles, but both algorithms run in $O(n)$ time, which is optimal for both problems as the complexity of the output alone can be $\Omega(n)$ \cite{Shermer}.

\section{Funnel polygons} \label{sec:funnel}

We describe polygons as a list of vertices $v_1, v_2 , . . . v_n$ where the index reflects the clockwise order of the vertices. We will also refer to the edges with the same indices, i.e. $e_i = v_{i}v_{i+1}$. We use the notion of ``left'' tests or orientation tests defined on three points $p_1, p_2, p_3$. $p_3$ is said to be to the ``left'' of $p_1,p_2$ if a left/counter-clockwise turn is made at $p_2$ in the polygonal chain $p_1,p_2,p_3$. Similarly, $p_3$ is said to be to the ``right'' of $p_1,p_2$ if a right/clockwise turn is made at $p_2$ in the polygonal chain $p_1,p_2,p_3$.

\textit{Funnel polygons} are defined as simple polygons composed of a convex edge $e_n = v_nv_1$ with two reflex chains $R_1 = v_1,v_2,...,v_t$ and $R_2 = v_t,...,v_{n-1},v_n$. Assume without loss of generality that $t \geq n-t$, i.e. $|R_1| \geq |R_2|$. We will also assume that each reflex vertex is strictly reflex, i.e. no three vertices within a chain are collinear and there is a turn at every vertex. If this is not the case, we can simply remove the vertices that are not strictly reflex. Funnel polygons are also sometimes referred to as \textit{tower polygons}.

We will use a recursive procedure to compute the hidden set and convex cover, and thus will be able to prove the correctness of the algorithm using induction. Our algorithm gives both a hidden set $H$ and a convex cover $C$, but can be implemented to only find one. The algorithm finds both of these in $O(n)$ time.

\begin{theorem}\label{thm:funnel}
For any funnel polygon $P$, a hidden set $H$ in $P$ and a convex cover $C$ of $P$ such that $|H| = |C|$ can be found in linear-time. 
\end{theorem}

\begin{proof}
Algorithm \ref{alg:funnel} essentially determines a convex decomposition (with Steiner points) such that for each piece in the decomposition, a hidden point can be placed along either $R_1$ or $R_2$ corresponding to it. 

The algorithm considers two cases. The first, easier case, is where all of the edges $e_i, e_{n-i}$ for $i \in [0,n-t]$ strongly see each other, i.e. they form a convex quadrilateral $v_i,v_{i+1},v_{n-i},v_{n-i+1}$ that is within the polygon $P$. We will refer to this as the Case 1 and the case in which this is not true as Case 2. We give two examples of Case 1 in Figure \ref{fig:funnel_easy}. If this is the case, then we can simply place a hidden point on each midpoint of $R_1$ (the longer chain) and connect the corresponding edges $e_i, e_{n-i}$ until there are no more edges to connect across with from $R_2$. At this point, we can simply partition the remaining region into triangles between the remaining edges of $R_1$ and the vertex $v_{t+1}$.

\begin{figure}[ht]
    \centering
    \includegraphics[width=.5\textwidth]{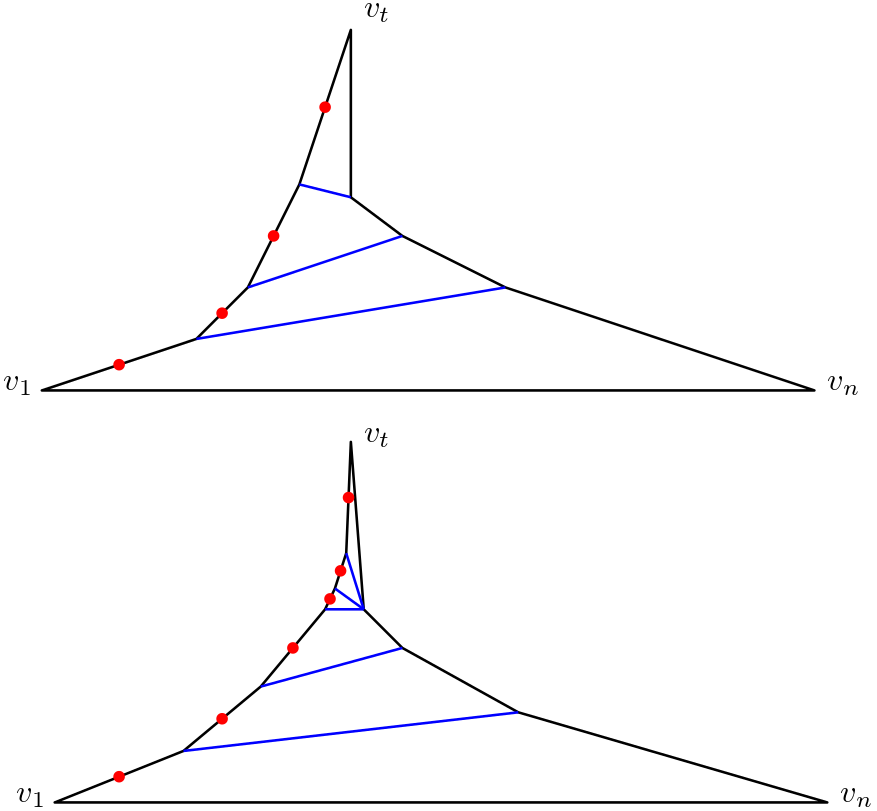}
    \caption{Case 1 for the funnel polygons.}
    \label{fig:funnel_easy}
\end{figure}

From here, we can use an induction on $n$, the number of vertices, to cover Case 2 where not all edge pairs satisfy the strong visibility requirement of Case 1. Clearly, when $n=3$, we have an instance of Case 1 and thus already have a solution. From here we just need to show that if it holds for $n \in [3,k]$, then it also holds for $n = k+1$. We can assume that for when $n=k+1$, we do not have an instance of Case 1 since we have already shown that. Let $e_i,e_{n-i}$ be the lowest $i$ pair for which $v_i,v_{i+1},v_{n-i},v_{n-i+1}$ is not within $P$. 

If $v_i$ does not see $v_{n-i+1}$, then the pair $e_{i-1} e_{n-i+1}$ would also not have their quadrilateral within $P$ and $i-1 < i$, so $v_i$ sees $v_{n-i+1}$. If $v_{i+1}$ does not see $v_{n-i}$, then either $v_i$ does not see $v_{n-i}$ or $v_{n-i+1}$ does not see $v_{i+1}$ because if there is an internal vertex $u$ in the shortest path from $v_{i+1}$ to $v_{n-i}$, it must be from either the chain of $R_1$ after $v_{i+1}$ or the chain of $R_2$ before $v_{n-i}$. This follows from \cite[Lemma 1]{GhoshWeakVisible}, since funnel polygons are a subclass of polygons weakly visible from a convex edge. Because $R_1$ is a reflex chain, if $u \in R_1$ then $u$ must also be in the shortest path from $v_i$ to $v_{n-i}$. If $u \in R_2$, then $u$ must be in the shortest path from $v_{n-i+1}$ to $v_{i+1}$. See Figure \ref{fig:funnel_block} for an example. Since these shortest paths have internal vertices, the endpoints cannot see each other.

\begin{figure}[ht]
    \centering
    \includegraphics[width = .5\textwidth]{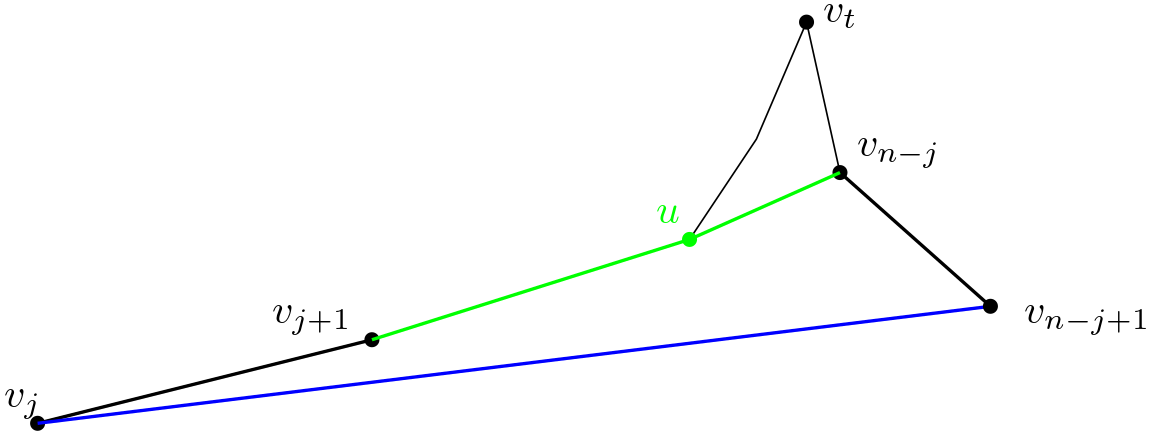}
    \caption{Lemma 1 of \cite{GhoshWeakVisible} implies that if $v_{i+1}$ does not see $v_{n-i}$ then either $v_i$ does not see $v_{n-i}$ or $v_{n-i+1}$ does not see $v_{i+1}$.}
    \label{fig:funnel_block}
\end{figure}

\begin{figure}[ht]
    \centering
    \includegraphics[width = .5\textwidth]{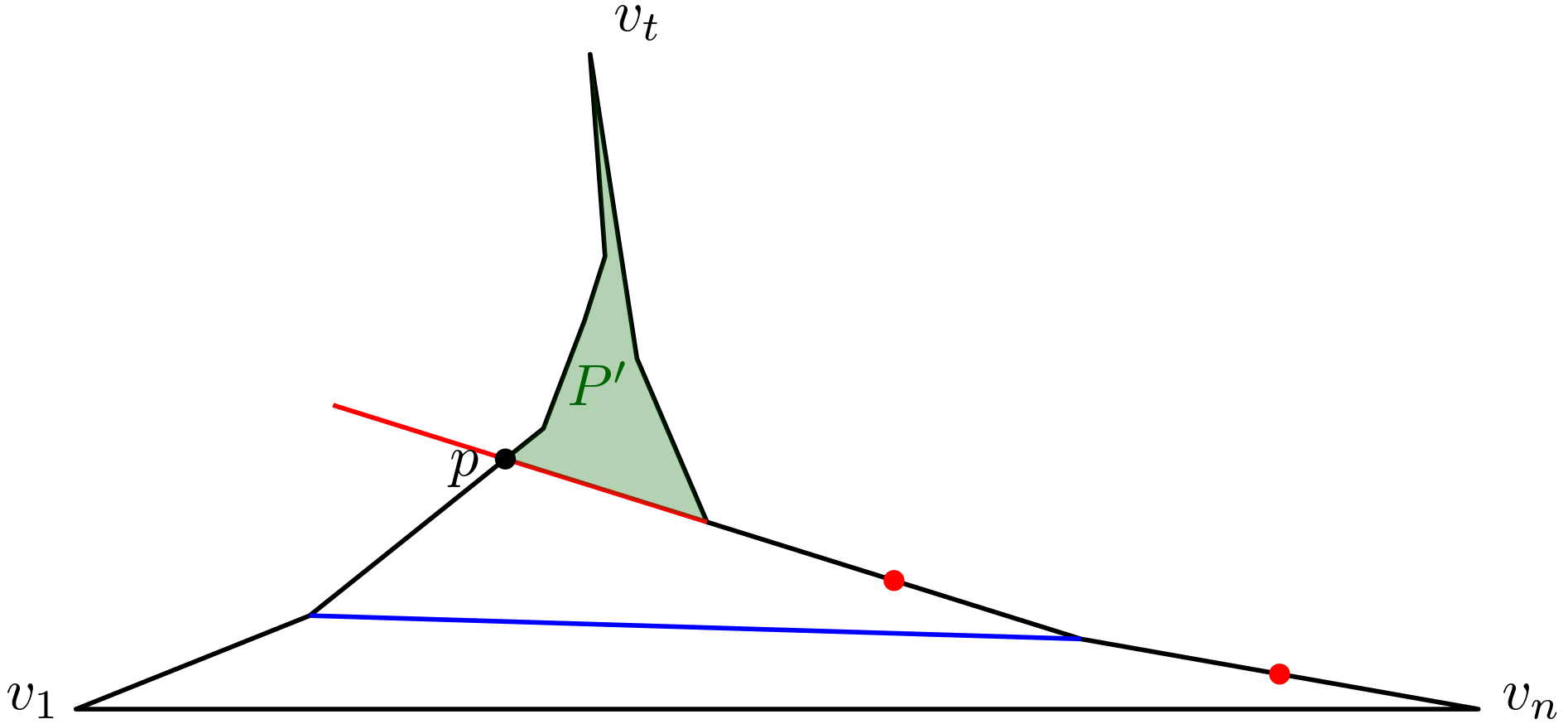}
    \caption{Case 2 for funnel polygons.}
    \label{fig:funnel_recurse}
\end{figure}

Therefore, since in all cases we find a convex cover and hidden set of the same size, Theorem \ref{thm:funnel} holds. Algorithm \ref{alg:funnel} runs in $O(n)$ time. First, note that determining whether we are in Case 1 or Case 2 only takes $O(1)$ time per edge in $P$. This is because the only way to be in Case 2 is if the edge extension of one edge intersects with another edge, which takes $O(1)$ to determine. We do not need to compare with any other edges of the polygon. $P'$ has complexity at most 1 less than $P$, which is only the case if the index $i$ at which we form $P'$ is 1, giving us the following recurrence:
$$T(n) = T(n-1) + O(1)$$

The base case is $T(1) = O(1)$, which solves to a total runtime of $O(n)$.

\end{proof}

\begin{algorithm}[H]
\caption{Algorithm for hidden set and convex cover in funnel polygons}\label{alg:funnel}
\KwData{$P$, a funnel polygon with convex edge $v_nv_1$}
\KwResult{$H$, a maximum hidden set in $P$ and $C$, a minimum convex cover of $P$}

$H \leftarrow \{\}$\;
$C \leftarrow \{\}$\;
$i \leftarrow 1$\;

\While{$i \leq t$} {
\If{$v_{n-i}$ is left of $v_iv_{i+1}$}
{
    \ForEach{$j \in [1,i]$}
    {
        Add the midpoint of $e_j$ to $H$\;
        Add quadrilateral $v_j,v_{j+1},v_{n-j},v_{n-j+1}$ to $C$\;
    }
    Determine the point $p$ of intersection between   $e_{n-i-2}$ and the extension of $e_i$\;
    Add quadrilateral $v_i,v_{i+1},p,v_{n-i+1}$ to $C$\;
    Recurse on the subpolygon $P' = v_{i+1}, ...v_{n-i},p$. This yields $(H',C')$ \; 
    Return $(H \cup H', C \cup C')$\; 
}
\If{$v_{i+1}$ is right of $v_{n-i+1}v_{n-i}$}
{
    \ForEach{$j \in [1,i]$}
    {
        Add the midpoint of $e_{n-j}$ to $H$\;
        Add quadrilateral $v_j,v_{j+1},v_{n-j},v_{n-j+1}$ to $C$\;
    }
    Determine the point $p$ of intersection between    $e_i$ and the extension of $e_{n-i-2}$\;
    Add quadrilateral $v_i,p,v_{n-i},v_{n-i+1}$ to $C$\;
    Recurse on the subpolygon $P' = p, v_{i+1} ...v_{n-i}$. This yields $(H',C')$ \; 
    Return $(H \cup H', C \cup C')$\; 
}
$i \leftarrow i+1$\;
}
\ForEach{$j \in [1,t]$}
    {
        Add the midpoint of $e_{j}$ to $H$\;
        \If {$j<n-t$} {
        Add quadrilateral $v_j,v_{j+1},v_{n-j},v_{n-j+1}$ to $C$\; }
        \Else {
        Add triangle $v_j,v_{j+1}, v_{t+1}$\;
        }
    }
\end{algorithm}

\section{Variant problems in funnel polygons} \label{sec:variants}

There are a few interesting properties of Algorithm~\ref{alg:funnel} which work even under more constrained versions of our problem. Specifically, our hidden points are constrained to the boundary of $P$ and there is no overlap between any of the convex pieces, implying that we have obtained a convex decomposition. For a general simple polygon, Chazelle and Dobkin \cite{chazelle1985optimal} gave an $O(n^3)$ algorithm for convex decomposition while allowing for the use of Steiner points (points that are not vertices of the polygon). Since this is a more constrained problem, naturally the minimum convex decomposition is at least as large as the minimum convex cover, meaning that our $O(n)$ solution is optimal for the problem of minimum convex decomposition in funnel polygons.

We also can permit additional constraints on the hidden set. In particular, our algorithm places points along the boundary of the polygon. This means that for funnel polygons, there is no ``advantage'' from using points in the interior of the polygon as the maximum hidden set is the same with or without constraining to the boundary. 

We can also adapt our algorithm to deal with the case of maximum hidden vertex set. Although the properties of the visibility graph from \cite{choi1995characterizing} imply a polynomial algorithm for hidden vertex set, we can achieve the same result in $O(n)$.

\begin{theorem}\label{thm:funnel_vertex}
For any funnel polygon $P$, a hidden vertex set $H$ in $P$ and a convex cover $C$ of the vertices of $P$ such that $|H| = |C|$ can be found in linear-time. 
\end{theorem}

\begin{proof}
Instead of using a convex cover of the whole polygon, we only need to find a convex cover of the vertices to bound the hidden vertex set. This is because every convex piece in the cover of the vertices can have at most 1 vertex in the hidden vertex set by the definition of convexity (all points see each other). We again break into 2 cases, one where no recursion is necessary and another where we can use recursion.

In Case 1, we again have that every edge $e_i$ can strongly see its corresponding partner $e_{n-i}$. If this is the case, then we can place hidden vertices on the longer chain, starting with $v_1$ and alternating until $v_t$ is reached. All odd vertices of index less than or equal to $t$ will be a hidden vertex. We can use the same cover as last time, but again alternating using those corresponding to the odd indexed edges before $e_t$. This is sufficient to cover the vertices as every even vertex $v_{2i}$ is in the same quadrilateral as $v_{2i-1}$ corresponding to edge $e_{2i-1}$. This quadrilateral will also cover $v_{n-2i}$, $v_{n-2i-1}$. See Figure \ref{fig:funnel_easy_vertex} for an illustration.

\begin{figure}[ht]
    \centering
    \includegraphics[width=.5\textwidth]{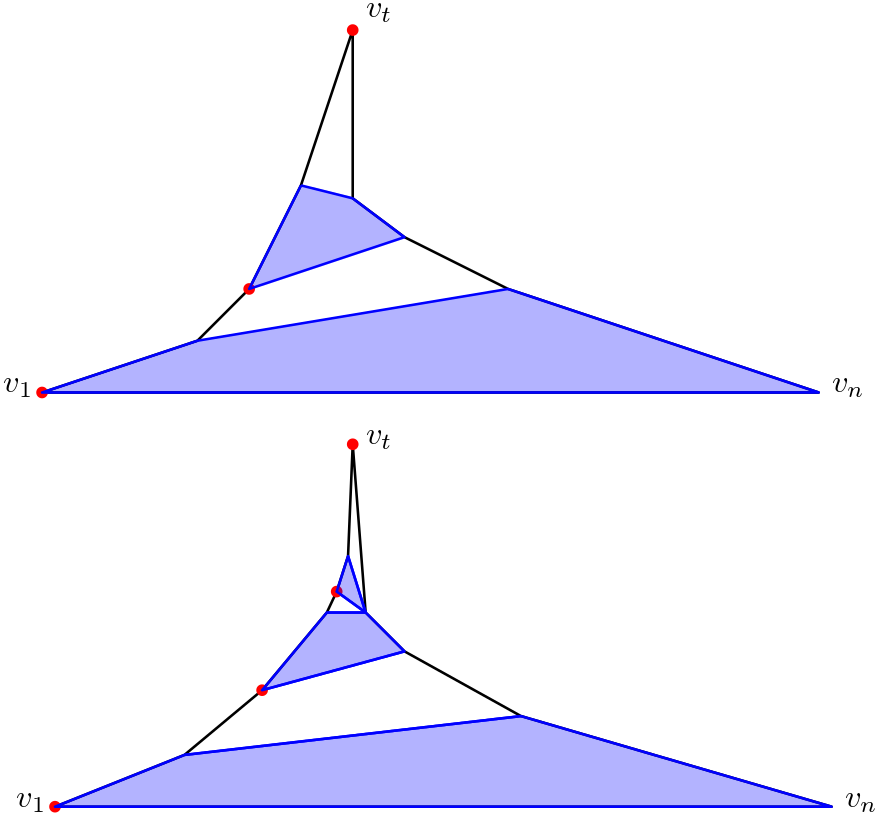}
    \caption{Case 1 for the funnel polygons with respect to hidden vertex set.}
    \label{fig:funnel_easy_vertex}
\end{figure}

Case 2 is identical to that of Theorem \ref{thm:funnel}, but it needs some slight care depending on whether the lowest pair $e_i e_{n-i}$ is such that $i$ (or $n-i$) is odd or even. Without loss of generality, assume that it is $v_i$ that cannot see $v_{n-i-1}$. If $i$ is even, then we will place at all odd numbered vertices less than $i$ and use all the odd numbered quadrilaterals, $v_j v_{j+1} v_{n-j-1}v_{n-j}$ for odd $j$ where $j < i$. We will then recurse on $P' = v_{i+1} ... v_{n-j-1}$. Note that $P'$ does not include any new vertices, and assume that the vertices are  renumbered to reflex their index in the subpolygon. This renumbering can be done implicitly by keeping track of the first and last vertex index in the current recursion layer. This case is depicted in the top picture of Figure \ref{fig:funnel_recurse_vertex}. If $i$ is odd, then we will repeat the procedure for odd numbered vertices and quadrilaterals less than $i$. However, we will also include a triangle, $v_{i}v_{i+1}v_{n-j}$. We then recurse on $P' = v_{i+2} ... v_{n-j-1}$. Again, no new vertices are made and all vertices outside $P'$ have been covered. This subcase is shown in the bottom figure of Figure \ref{fig:funnel_recurse_vertex}.

\begin{figure}[ht]
    \centering
    \includegraphics[width=.5\textwidth]{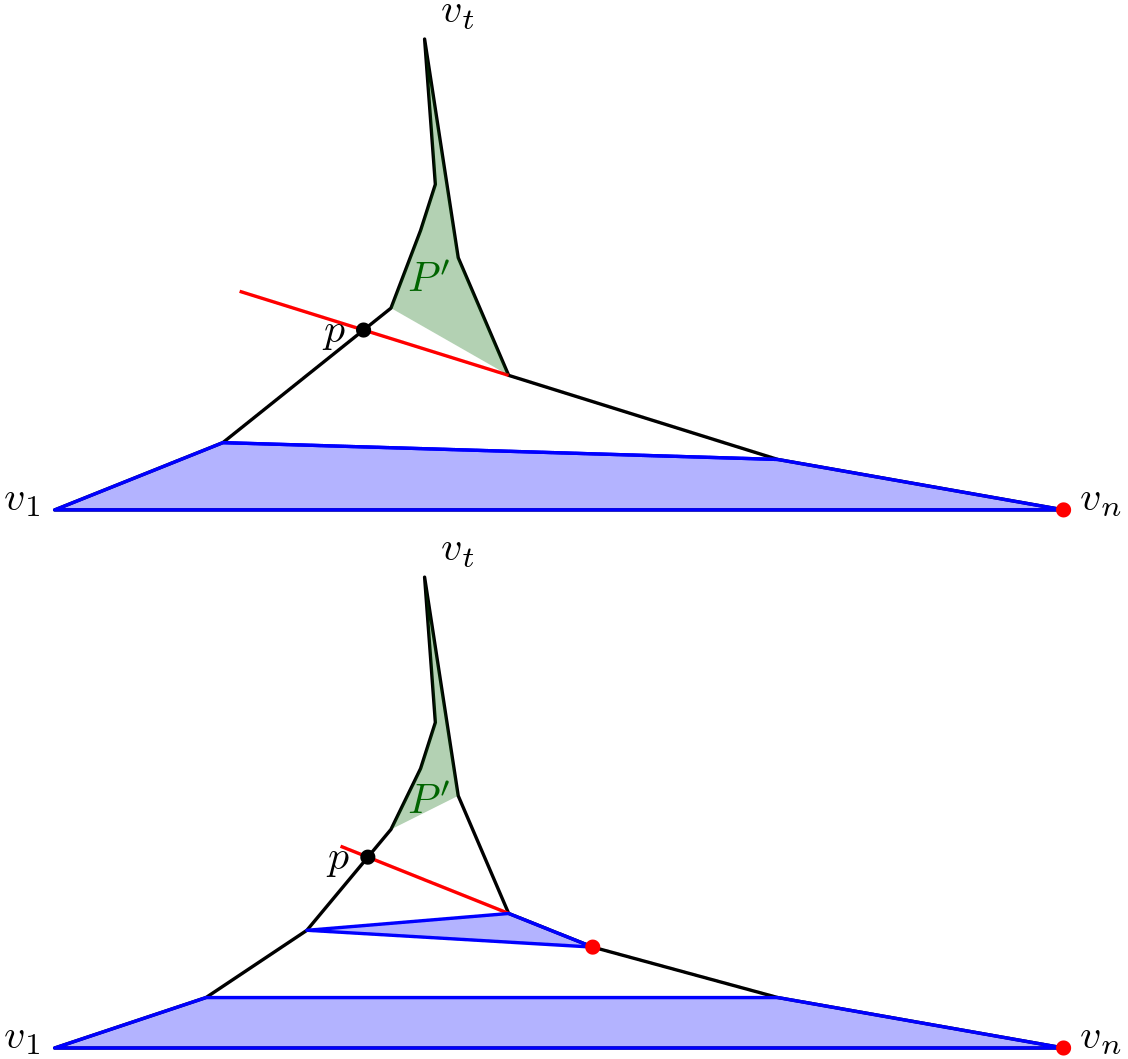}
    \caption{Case 2 for the funnel polygons with respect to hidden vertex set.}
    \label{fig:funnel_recurse_vertex}
\end{figure}

Therefore, by induction, we have a hidden vertex set and a convex cover of the vertices of the same size. All that remains is the recursion for time complexity. In the worst case, $P'$ has 2 less vertices than $P$, which means that we have:

$$T(n) = T(n-2) + O(1)$$

The base case is $T(1) = O(1)$, which solves to $T(n) = O(n)$. 
\end{proof}

\section{Pseudotriangles} \label{sec:pseudos}

Pseudotriangles are defined as simple polygons containing exactly 3 convex vertices, which means that it can be represented as three reflex chains $R_1 = v_1,v_2,...,v_t$, $R_2 = v_t, ... v_{s-1}, v_s$, $R_3 = v_s, . . .  v_{n-1}, v_n$. The three convex vertices here are $v_1, v_t,$ and $v_s$. Without loss of generality, we will consider the cardinalities of the chains to be such that $R_1 \geq R_2 \geq R_3$. Funnel polygons are a subclass of pseudotriangles, where $|R_3| = 1$. We will only consider pseudotriangles that are not funnel polygons, as Section \ref{sec:funnel} already gives a linear-time algorithm for funnel polygons.

\begin{theorem}\label{thm:pseudo}
For any pseudotriangle $P$, a hidden set $H$ in $P$ and a convex cover $C$ of $P$ such that $|C| \leq 2|H|$, i.e. a 2-approximation, can be found in linear-time .
\end{theorem}

\begin{proof}

We show that by using the algorithm for funnel polygons in Section \ref{sec:funnel}, we can achieve a linear-time 2-approximation for pseudotriangles. The procedure is simple: Consider the edge $e_1 = v_1v_2$, and its extension. The extension of $e_1$ will land at some point $p$ on $R_2$ called $p$. It cannot land on either $R_1$ or $R_3$ otherwise the landed on chain would not be reflex, and it must intersect some part of the boundary of $P$. Let $p$ rest on edge $e_i = v_i v_{i+1}$. This partitions $P$ into two subpolygons $P_1 = v_2,v_3 ...v_i, p$ and $P_2 = p, v_{j+1}, ... v_1$. Both $P_1$ and $P_2$ are composed of two reflex chains and a convex edge, making them both funnel polygons. This is depicted in Figure \ref{fig:pseudo_example}.

\begin{figure}[ht]
    \centering
    \includegraphics[width=.3\textwidth]{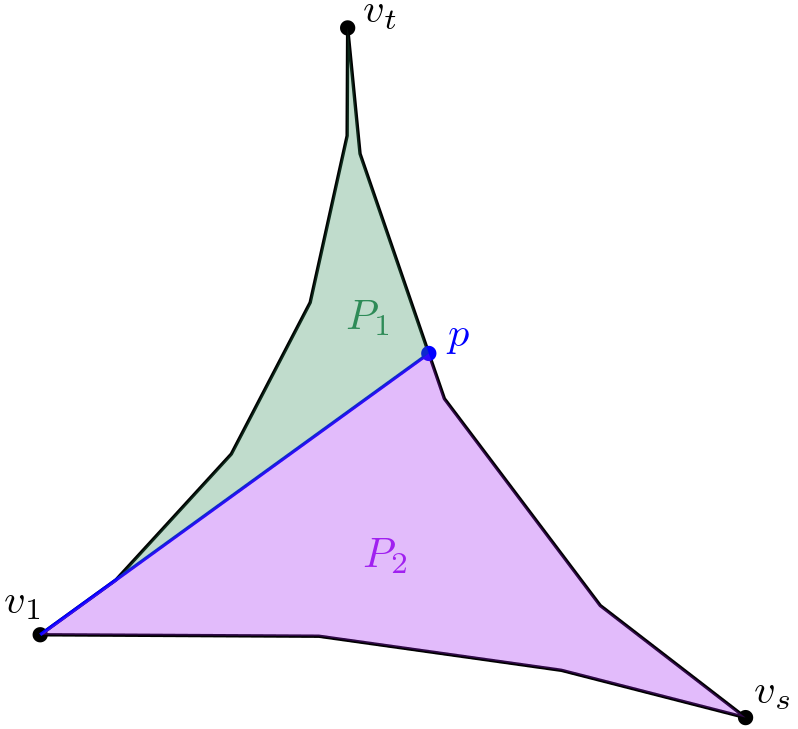}
    \caption{Pseudotriangle partition into funnel polygons that allow for a 2-approximation.}
    \label{fig:pseudo_example}
\end{figure}

From Theorem \ref{thm:funnel}, we know that we can obtain from $P_1$ a hidden set $H_1$ and a convex cover $C_1$. By symmetry, $P_2$ we get $H_2$ and $C_2$. As a hidden set, we take the larger, $H = \text{max}(H_1,H_2)$. As a convex cover, we take the union, $H = C_1 \cup C_2$. Because $|H_1| = |C_1|$ and $|H_2| = |C_2|$, we know that $|C| \leq 2 |H|$. This proves our claim, and since finding $p$ takes at most $O(n)$ time, as does each funnel, the total runtime is $O(n)$.
\end{proof}

Note, that while this algorithm does not achieve an exact answer, it is impossible to have an algorithm that only places hidden points on the edges of a pseudotriangle and that finds the exact maximum hidden set. The famous GFP (Godfried's Favorite Polygon) counterexample suffices here, and we depict it in Figure \ref{fig:godfried_example}. We give a convex cover of the boundary of size 3, which means that any hidden set constrained to the boundary must have size at most 3. We also give a hidden set of size 4, which means that 3 pieces cannot suffice for a convex cover.

\begin{figure}[ht]
    \centering
    \includegraphics[width=.3\textwidth]{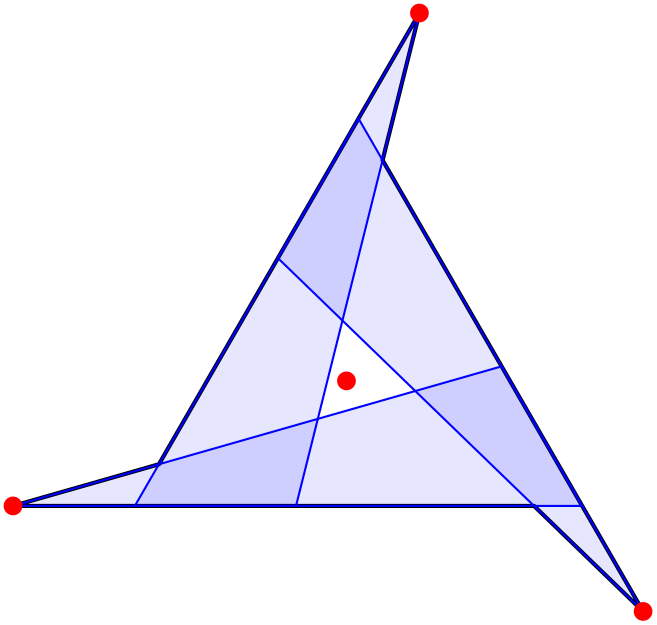}
    \caption{Godfried's favorite polygon, with a convex cover of the boundary of size 3 and a hidden set of size 4.}
    \label{fig:godfried_example}
\end{figure}

Therefore, any algorithm that places hidden points only on the boundary of a pseudotriangle cannot achieve an approximation factor better than $4/3$. 

\begin{theorem} \label{thm:pseudo_vertex}
For any pseudotriangle $P$, a hidden vertex set $H$ in $P$ and a convex cover $C$ of the vertices of $P$ such that $|C| \leq 2 |H|$ can be found in linear-time. 
\end{theorem}
\begin{proof}

    For the vertex version, we again split $P$ into two funnel polygons $P_1,P_2$ using the point $p$ (intersection of the edge extension of $e_1$ with the boundary). Using Theorem \ref{thm:funnel_vertex}, this will give us two hidden sets, $H_1$ and $H_2$ as well as two convex covers of the vertices $C_1$ and $C_2$. We can again take the union of $C_1$ and $C_2$ and the larger of $H_1$ or $H_2$, but we have some issues regarding the point $p$ that was used in the partition. Both $C_1$ and $C_2$ will cover $p$, so if we wish to constrain our convex pieces to vertices, we can easily clip out $p$ from the pieces that it is in. The larger issue is with the hidden sets. 
    
    To make $H_1$ into a hidden vertex set, we need to determine whether the point $p$ is in $H_1$. If it is not, we do nothing as $p$ is the only non-vertex point in $P_1$. If it is, then we will replace $p$ with $v_{i+1}$ ($p$ lies on edge $e_i = v_iv_{i+1}$). We know that $v_{i+1}$ cannot see any hidden point in $H_1$ because the only vertices that $p$ sees in $P_1$ are $v_i$ and those in the chain $R_1$ that are below the extension of edge $p v_i$. This is a superset of the vertices in $P_1$ that $v_{i+1}$ sees, which follows from considering the funnel $v_1, v_2, ... v_i v_{i+1}$. In this funnel, $v_{i+1}$ can only see those vertices under ``before'' the edge extension of $e_i$. We give an example of this shift to $v_{i+1}$ in Figure \ref{fig:pseudo_vertex}.

    \begin{figure}[ht]
    \centering
    \includegraphics[width=.6\textwidth]{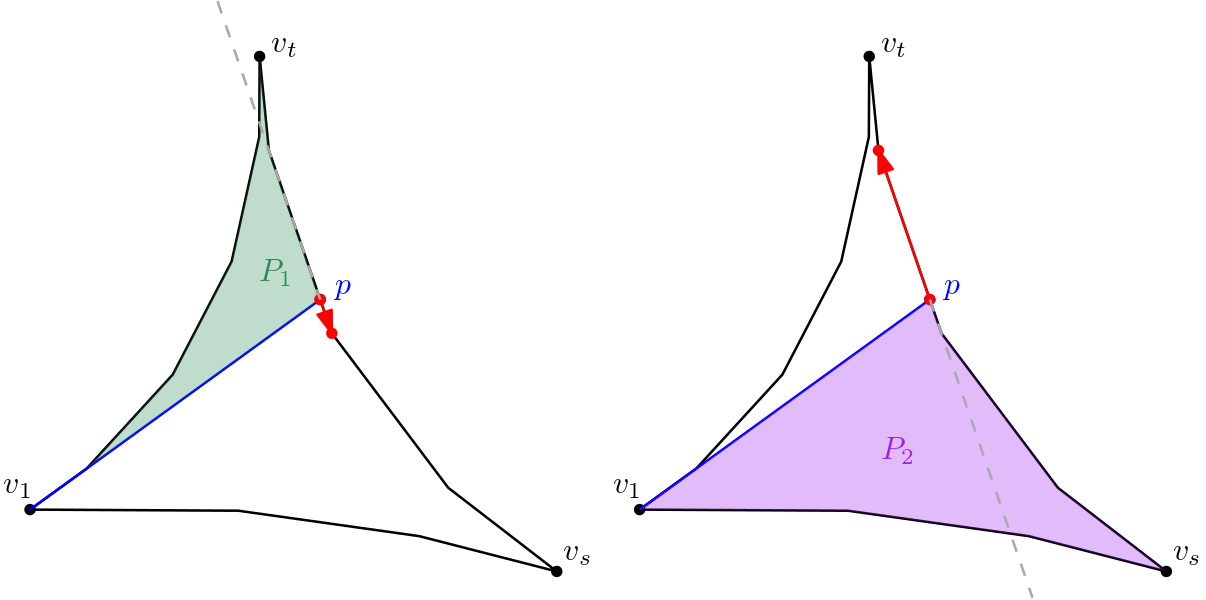}
    \caption{Shifting a non-vertex hidden point $p$ to a vertex. Left: shifting to a vertex in $P_1$ Right: shifting to a vertex in $P_2$.}
    \label{fig:pseudo_vertex}
    \end{figure}

    The same can be done for $H_2$, but instead shift a hidden vertex placed at $p$ to $v_i$. By symmetry, the exact same argument applies. Therefore, we can take the adjusted hidden sets $H_1'$ and $H_2'$ and use the larger of the two as our hidden vertex set $H$. Since $|H_1'| = |H_1| = |C_1|$ and $|H_2'| = |H_2| = |C_2|$, we again have that $|C| \leq 2|H| $, hence proving our claim. Finding $p$ takes naively $O(n)$ time and running the funnel algorithm on both $P_1$ and $P_2$ is at most $O(n)$. This gives us a total runtime of $O(n)$.
    
\end{proof}

\section{Conclusion}
We show that funnel polygons are homestead polygons. We do so by giving linear-time algorithms for both hidden set and convex cover. This, combined with the results from Browne and Chiu \cite{browne2022collapsing}, give us 3 classifications of polygon for which the hidden set number and convex cover number both coincide and can be calculated in linear-time: spiral polygons, histogram polygons, and funnel polygons. We also now have a class of polygons for which minimum convex cover, maximum hidden set, and maximum hidden vertex set can be approximated within a factor of 2, pseudotriangles. 

We believe that the approximation factor for pseudotriangles can be greatly reduced, and do not believe the problems to be NP-hard for the pseudotriangle case. We also believe that while interior hidden points may be necessary for some instances, we suspect that there is no need to include more than 1, as is needed in the Godfried example. Naturally more open questions revolve around looking at the problems for more polygonal subclasses. We are particularly interested in monotone mountains and convex fans and believe similar techniques may be used for solving or at least approximating in those cases.


\bibliographystyle{plainurl}

\bibliography{credits}

\newpage

\end{document}